\documentclass[letterpaper, conference]{IEEEtran}

\usepackage{amsmath,amssymb,amsthm,mathtools,graphicx,booktabs,bm}
\usepackage{xcolor}
\usepackage{cite,url}
\usepackage{microtype}

\newtheorem{proposition}{Proposition}
\newtheorem{corollary}{Corollary}
\newtheorem{remark}{Remark}

\newcommand{\R}{\mathbb{R}}
\newcommand{\X}{\mathcal{X}}
\newcommand{\Xd}{\mathcal{X}_d}
\newcommand{\Xt}{\mathcal{X}_{\mathrm{true}}}
\newcommand{\Hset}{\mathcal{H}}
\newcommand{\one}{\mathbf{1}}
\newcommand{\tran}{^{\top}}
\newcommand{\norm}[1]{\lVert #1 \rVert}
\newcommand{\abs}[1]{\lvert #1 \rvert}
\newcommand{\diam}{\operatorname{diam}}
\newcommand{\widthop}{\operatorname{width}}
\newcommand{\vol}{\operatorname{vol}}

\title{Geometry-Aware Set-Membership Multilateration:\\ Directional Bounds and Anchor Selection
\vspace{.2cm}}

\author{%
\IEEEauthorblockN{Giuseppe C. Calafiore} 
\IEEEauthorblockA{{\em Fellow}, IEEE -- Department of Electronics and Telecommunications, Politecnico di Torino, Italy\\
Email: giuseppe.calafiore@polito.it}
\vspace{.5cm}}

\begin{document}
\maketitle

\begin{abstract}
In this paper, we study anchor selection for range-based localization under unknown-but-bounded measurement errors.
We start from the convex localization set $\X=\Xd\cap\Hset$ recently introduced in \cite{CalafioreSIAM}, where $\Xd$ is a polyhedron obtained from pairwise differences of squared-range 
 equations between the unknown location $x$ and the anchors, and $\Hset$ is the intersection of  upper-range hyperspheres.
Our first goal is \emph{offline} design: we derive geometry-only E- and D-type scores from the centered scatter matrix $S(A)=AQ_mA\tran$, where $A$ collects the anchor coordinates and $Q_m=I_m-\frac{1}{m}\one\one\tran$ is the centering projector,  showing that $\lambda_{\min}(S(A))$ controls worst-direction and diameter surrogates for the polyhedral certificate $\Xd$, while $\det S(A)$ controls principal-axis volume surrogates.
Our second goal is \emph{online} uncertainty assessment for a selected subset of anchors: exploiting the special structure $\X=\Xd\cap\Hset$, we derive a simplex-aggregated enclosing ball for $\Hset$ and an exact support-function formula for $\Hset$, which lead to finite hybrid bounds for the actual localization set $\X$, even when the polyhedral certificate deteriorates.
Numerical experiments are performed in two dimensions,  showing that geometry-based subset selection is close to an oracle combinatorial search, that the D-score slightly dominates the E-score for the area-oriented metric considered here, and that the new $\Hset$-aware certificates track the realized size of the selected localization set closely.
\end{abstract}

\section{Introduction}
Range-based localization and multilateration are classical primitives in navigation, robotics, and sensor networks.
Standard formulations typically assume stochastic errors and seek one point estimate via least-squares, maximum-likelihood, or convex-relaxation methods;
representative references include the classical linearized estimator of Chan and Ho, the source-localization formulation of Beck, Stoica, and Li, and SDP/SOCP relaxations for sensor-network localization \cite{ChanHo1994,BeckStoicaLi2008,Biswas2006,Tseng2007}.
Classical algebraic multilateration formulations for GPS-type equations go back to Bancroft \cite{Bancroft1985}, while robust convex formulations also address anchor-position uncertainty \cite{NaddafzadehShirazi2014}.
A recent survey is given in \cite{WiddisonLong2024}. 
By contrast, set-membership and interval methods seek guaranteed uncertainty sets under unknown-but-bounded errors, and are classical in bounded-error identification and interval analysis \cite{Milanese1996,Jaulin2001}.

In the journal precursor \cite{CalafioreSIAM}, multilateration was addressed  in this latter deterministic framework.
The range errors are assumed only to satisfy interval bounds, and the goal is to compute a guaranteed set containing all positions compatible with the measurements.
The key step in \cite{CalafioreSIAM} is a difference-of-measurements construction that eliminates the quadratic term and yields a convex polyhedron $\Xd$ containing the true nonconvex feasible set $\Xt$.
Intersecting $\Xd$ with the upper-ball set $\Hset$ gives a refined localization set
\begin{equation}
\X = \Xd \cap \Hset, \qquad \Xt \subseteq \X \subseteq \Xd.
\label{eq:Xintro}
\end{equation}
The same paper also develops SOCP/SDP procedures for inner and outer approximations of $\X$.
Among the papers closest in spirit, Shi et al.\ \cite{ShiMaoAndersonYangChen2017} also assume bounded range errors, but formulate a worst-case point-estimation problem solved via semidefinite relaxation rather than a global set-valued certification problem.
One of the open directions explicitly pointed out in \cite{CalafioreSIAM} is the need for a formal analysis of how conservative the inclusion $\Xt \subseteq \X$ becomes as the anchor geometry deteriorates.
This paper takes a design-and-selection viewpoint on that question. Suppose a platform has access to a pool of candidate anchors or ranging agents, but can activate only a subset because of communication, power, or scheduling constraints.
Which subset should be chosen \emph{before} the measurements are collected, and how should its realized uncertainty be certified \emph{after} the measurements are available?
Our answer separates these two tasks. For \emph{offline} design, we derive geometry-only E- and D-type scores from the centered scatter matrix $S(A)$, leading to target-free rules for anchor subset selection.
For \emph{online} certification, we exploit the special structure $\X=\Xd\cap\Hset$ and derive new $\Hset$-aware size bounds that explain why the actual localization set can remain much smaller than the loose polyhedral certificate suggested by $\Xd$ alone.
The resulting picture is a deterministic, set-membership analogue of anchor-geometry design in stochastic localization.
The closest approach in spirit on the stochastic side is the literature on Cram\'er--Rao/Fisher-information bounds, sensor placement, and geometry design for range-only localization, see, e.g., \cite{ChangSahai2006,Bishop2010,MorenoSalinas2013};
the present contribution differs in that the objective is a guaranteed global uncertainty set rather than a local covariance surrogate.
Our new contributions are: 1) geometry-only E- and D-type anchor-selection scores with explicit worst-direction and principal-axis interpretations for the polyhedral certificate $\Xd$;
2) $\Hset$-aware online certificates that remain finite and informative even when $\Xd$ becomes ill-conditioned;
and 3) exhaustive two-dimensional experiments that separate offline selection from post-measurement certification and show near-oracle behavior on the tested instances.

\section{Set-membership localization problem}
Consider $m$ known anchors' positions $a^{(i)}\in\R^n$, $i=1,\dots,m$, collected in the matrix
\[
A=[a^{(1)}\ \cdots\ a^{(m)}]\in\R^{n\times m},
\]
and measurements $\xi_i$ are taken of the (squared) Euclidean distance between an unknown location $x$ and the anchors:
\begin{equation}
\norm{x-a^{(i)}}_2^2 = \xi_i,  \quad i=1,\dots,m.
\label{eq:model}
\end{equation}
Here, according to the setup proposed in \cite{CalafioreSIAM}, we assume that $\xi_i$ is corrupted by unknown-but-bounded noise, hence only interval bounds $\xi_i^-\le \xi_i\le \xi_i^+$ are assumed known.
This model actually covers both plain and squared range measurements with absolute or relative unknown-but-bounded errors, as detailed in \cite{CalafioreSIAM}.

\begin{remark}[Noise Model Equivalence]\rm In practice, depending on the physical principles they are based on (e.g. time-of-flight, power attenuation, etc.), different types of 
range measurement devices may have different 
types of errors, such as 
 errors on the distance itself, errors on the squared distance, or  {\em relative} errors.
However, 
under the interval uncertainty model, an interval error on the distance can be mapped directly and exactly to an equivalent model with interval uncertainty on the squared distance, and the same is true for relative errors on the distance or squared distance.
Indeed, Section~2 of \cite{CalafioreSIAM}) shows that the squared interval model  (\ref{eq:model}) 
is general and representative of most practical situations. 
\end{remark}

According to such model and bounds, the unknown target's position belongs to a nonconvex geometrical region $\Xt$, which is the intersection of $m$ hyperspherical annuli: 
\begin{equation}
\Xt = \bigcap_{i=1}^m
\{x\in\R^n: \xi_i^- \le \norm{x-a^{(i)}}_2^2 \le \xi_i^+\},
\label{eq:Xtrue}
\end{equation}
see Figure~\ref{fig_3anchors} for an illustrative example in dimension $n=2$ with $m=3$ anchors (and also Figure~\ref{fig:method} for a detailed visualization of the localization construction).

\begin{figure}[h!tb]
\begin{center}
\includegraphics[width=.4\textwidth]{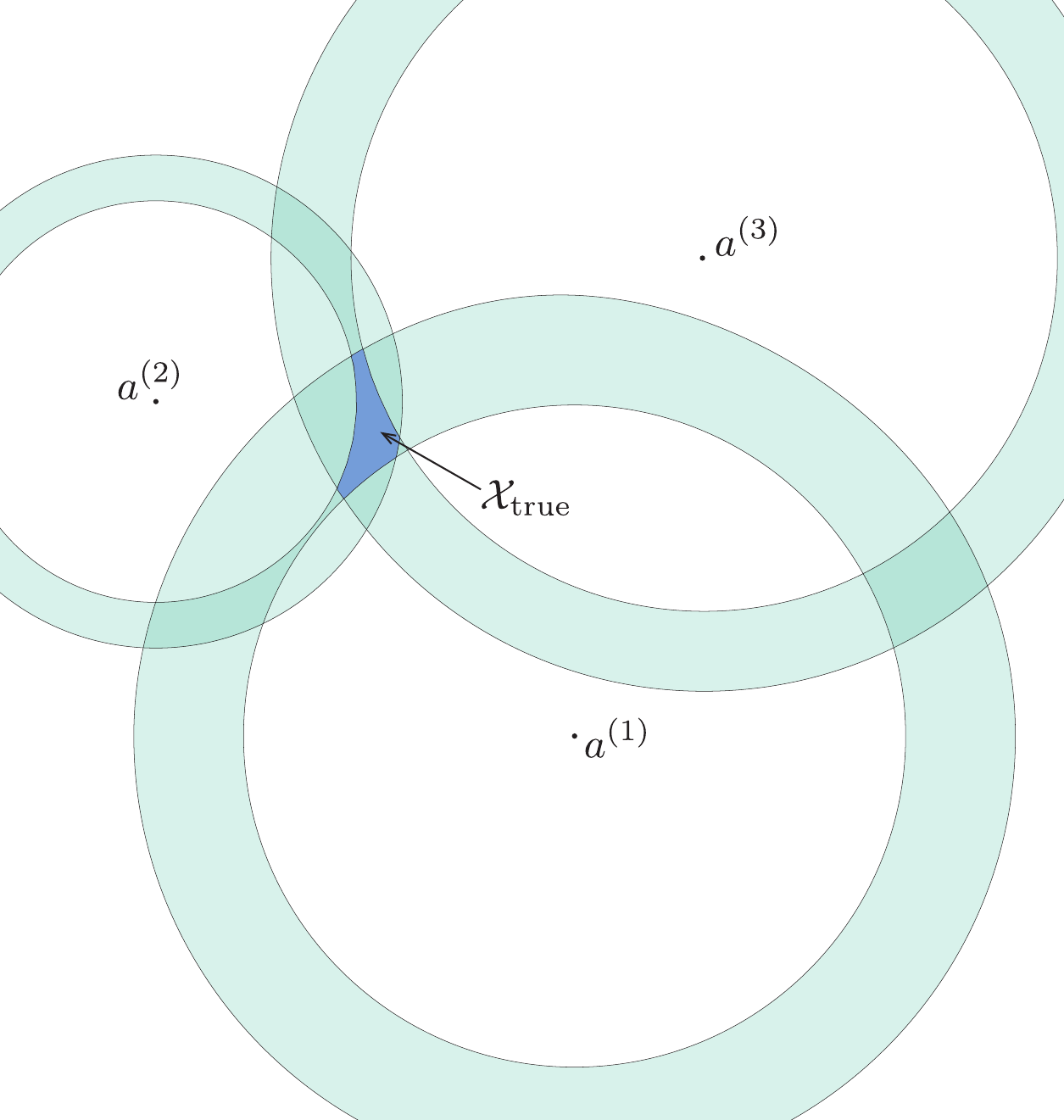}
  \caption{Localization with three anchors in $\R^2$. \label{fig_3anchors}}
\end{center}
\end{figure}

Now, for any pair of measurements $1\le i<j\le m$
we consider the difference of measurements equations \eqref{eq:model}, obtaining
\[
2(a^{(j)}-a^{(i)})\tran x = (\omega_j-\omega_i) - (\xi_j-\xi_i),
\]
which are affine in $x$, and in which we defined $\omega_i\doteq \norm{a^{(i)}}_2^2$ for all $i$.
Stacking all $p={m(m-1)}/{2}$ pairwise differences gives
\[
2EA\tran x = E(\omega-\xi),
\]
where $E\in\R^{p\times m}$ has rows $(e^{(i)}-e^{(j)})\tran$ for pairs $(i,j)$, $i<j$.
Equivalently (see \cite{CalafioreSIAM}), $E\tran$ is the incidence matrix of the complete {\em difference-of-measurements graph}, and it holds that $E\tran E = mQ_m$, where
$Q_m$ is the centering projector
\[
Q_m \doteq I_m - \frac{1}{m}\one\one\tran.
\]
 Let
\[
H=[\xi^-,\xi^+] = \{\xi\in\R^m: \xi^-\le\xi\le\xi^+\},
\]
 denote the measurement hyperrectangle. Proposition~3 in \cite{CalafioreSIAM} shows that
the set of all $x$ compatible with all difference-of-measurements equations for some $\xi\in H$ is a polyhedron which admits the explicit representation
\begin{equation}
\begin{aligned}
\Xd = \bigl\{x\in\R^n: \exists\,\alpha\in\R,\;&\xi^- \le Q_m\omega \\
&{}- 2Q_mA\tran x + \alpha\one \le \xi^+ \bigr\}.
\end{aligned}
\label{eq:Xd}
\end{equation}
Moreover, if
\begin{equation}
\Hset = \bigcap_{i=1}^m
\left\{x\in\R^n:\norm{x-a^{(i)}}_2 \le \sqrt{\xi_i^+}\right\},
\label{eq:Hset}
\end{equation}
then the localization set is $\X=\Xd\cap\Hset$ and satisfies $\Xt\subseteq\X\subseteq\Xd$, see \cite{CalafioreSIAM}.
A useful consequence of \eqref{eq:Xd} is that the scalar $\alpha$ can be eliminated explicitly.
Let $q_i\tran$ denote row $i$ of $Q_mA\tran$, and, for a fixed $x$, define the intervals
\[
I_i(x) \doteq [\ell_i(x),u_i(x)],
\]
with
\[
\begin{aligned}
\ell_i(x) &\doteq \xi_i^- - (Q_m\omega)_i + 2q_i\tran x,\\
 u_i(x) &\doteq \xi_i^+ - (Q_m\omega)_i + 2q_i\tran x.
\end{aligned}
\]
Then \eqref{eq:Xd} is equivalent to the existence of one scalar $\alpha$ such that $\alpha\in I_i(x)$ for all $i$, that is,
\[
\bigcap_{i=1}^m I_i(x) \neq \varnothing.
\]
Since the $I_i(x)$ are intervals on the real line, a common intersection exists if and only if every lower endpoint is no larger than every upper endpoint, namely
\[
\ell_i(x) \le u_j(x), \qquad i,j=1,\ldots,m.
\]
Substituting the expressions above yields
\[
2(q_i-q_j)\tran x
\le
\xi_j^+ - \xi_i^- - (Q_m\omega)_j + (Q_m\omega)_i.
\]
Defining
\[
\beta_{ij} \doteq \xi_j^+ - \xi_i^- - (Q_m\omega)_j + (Q_m\omega)_i,
\]
we obtain the explicit halfspace representation
\begin{equation}
\Xd = \{x\in\R^n: 2(q_i-q_j)\tran x \le \beta_{ij},\;
i,j=1,\ldots,m\}.
\label{eq:Xdhalf}
\end{equation}
In other words, $\alpha$ only plays the role of a common shift aligning the $m$ scalar intervals $I_i(x)$;
eliminating $\alpha$ amounts to enforcing pairwise overlap of those intervals. Representation \eqref{eq:Xdhalf} is convenient in the optimization problems of Section~V.
The outer bounding set $\Xd$, and its refinement $\X$, provide set-membership outer-bounding estimates of the true set $\Xt$ of points compatible with the measurements.
In what follows we analyze the  dependence of $\Xd$ and  of $\X$ from the anchors' geometry.

\begin{figure*}[t]
\centering
\includegraphics[width=0.96\textwidth]{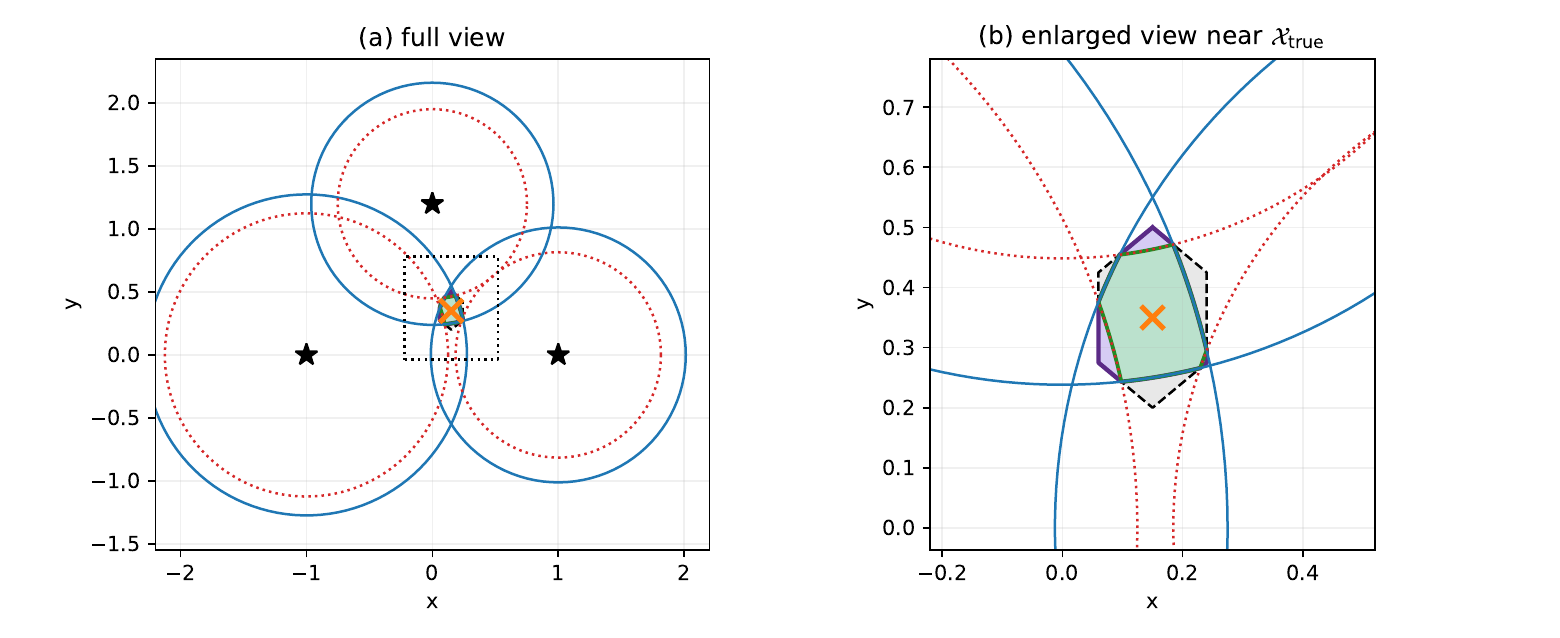}
\caption{Illustration of the localization construction on a three-anchor example. For visual clarity, this figure uses larger interval widths than the numerical experiments of Section~V, so that the annuli and the nested sets are easier to distinguish.
In the left panel, the blue solid circles and red dotted circles are the upper and lower range boundaries, the stars are the anchors, and the orange cross is the true target.
The black dashed polygon is $\Xd$, the purple region is $\X=\Xd\cap\Hset$, and the green region is the true feasible set $\Xt$.
The right panel enlarges the neighborhood of $\Xt$.}
\label{fig:method}
\end{figure*}

\section{Geometry scores and $\Hset$-aware certificates}
A relevant object in our analysis is the centered scatter matrix
\begin{equation}
S(A) \doteq AQ_mA\tran = \sum_{i=1}^m (a^{(i)}-\bar a)(a^{(i)}-\bar a)\tran,
\label{eq:S}
\end{equation}
where $\bar a = \frac{1}{m}\sum_{i=1}^m a^{(i)}$ is the anchors' centroid.
The matrix $S(A)$ is translation invariant, and it is positive definite exactly when the anchors do not lie in a proper affine subspace of $\R^n$.
Let
\[
w \doteq \xi^+ - \xi^- \in \R^m
\]
be the vector of interval widths.
For a compact set $\Omega\subset\R^n$ and a unit vector $v\in\R^n$, define
\[
\widthop_v(\Omega) \doteq \max_{x\in\Omega} v\tran x - \min_{x\in\Omega} v\tran x.
\]
When $S(A)\succ0$ and $\Xd\neq\varnothing$, the argument below also shows that $\Xd$ is bounded;
hence the directional widths of both $\Xd$ and $\X$ are finite in the nondegenerate case of interest.

\begin{proposition}
\label{prop:dirbound}
Assume $S(A)\succ0$.
Then for every unit vector $v\in\R^n$,
\begin{align}
\widthop_v(\X)
&\le \widthop_v(\Xd) \nonumber\\
&\le \frac{1}{2} w\tran \big|Q_mA\tran S(A)^{-1}v\big| \nonumber\\
&\le \frac{1}{2}\norm{w}_2\sqrt{v\tran S(A)^{-1}v}.
\label{eq:dirbound}
\end{align}
\end{proposition}

\begin{proof}
The first inequality follows from $\X\subseteq\Xd$. Pick any $x_1,x_2\in\Xd$. By \eqref{eq:Xd}, there exist $\alpha_1,\alpha_2\in\R$ and $\xi_1,\xi_2\in H$ such that
\[
\xi_\ell = Q_m\omega - 2Q_mA\tran x_\ell + \alpha_\ell\one,
\qquad \ell=1,2.
\]
Applying $Q_m$ to both sides and subtracting gives
\[
2Q_mA\tran(x_1-x_2) = Q_m(\xi_2-\xi_1).
\]
Multiplying on the left by $A$ and using $AQ_mA\tran=S(A)$ yields
\[
x_1-x_2 = \frac{1}{2}S(A)^{-1}AQ_m(\xi_2-\xi_1).
\]
Let $\delta=\xi_2-\xi_1$.
Since $\xi_1,\xi_2\in H$, we have $|\delta_i|\le w_i$ for all $i$, hence
\[
|v\tran(x_1-x_2)|
\le \frac{1}{2}w\tran \big|Q_mA\tran S(A)^{-1}v\big|.
\]
Taking the supremum over $x_1,x_2\in\Xd$ proves the second inequality in \eqref{eq:dirbound}.
The last inequality follows from Cauchy-Schwarz and
\begin{equation}
\begin{aligned}
\norm{Q_mA\tran S(A)^{-1}v}_2^2
&= v\tran S(A)^{-1}AQ_mA\tran S(A)^{-1}v \\
&= v\tran S(A)^{-1}v .
\end{aligned}
\label{eq:quadnorm}
\end{equation}
\end{proof}

\begin{corollary}
\label{cor:diam}
Under the assumptions of Proposition~\ref{prop:dirbound},
\begin{equation}
\diam(\X) \le \diam(\Xd) \le
\frac{\norm{w}_2}{2\sqrt{\lambda_{\min}(S(A))}}.
\label{eq:diam}
\end{equation}
\end{corollary}

\begin{proof}
Take the supremum of the last inequality in \eqref{eq:dirbound} over all unit vectors $v$, and use
\[
\max_{\norm{v}_2=1} v\tran S(A)^{-1}v
= \lambda_{\max}(S(A)^{-1})
= \frac{1}{\lambda_{\min}(S(A))}.
\]
\end{proof}

\begin{proposition}
\label{prop:ball}
Assume $\Hset\neq\varnothing$. For any simplex weight
\[
p\in\Delta_m \doteq \{p\in\R^m: p\ge0,\ \one\tran p=1\},
\]
define
\[
c(p) \doteq Ap
\]
and
\begin{equation}
\begin{aligned}
\rho(p)^2
&\doteq p\tran(\xi^+-\omega) + \norm{Ap}_2^2 \\
&= \sum_{i=1}^m p_i\xi_i^+
-\frac{1}{2}\sum_{i,j=1}^m p_i p_j
\norm{a^{(i)}-a^{(j)}}_2^2 .
\end{aligned}
\label{eq:rhop}
\end{equation}
Then
\begin{equation}
\X \subseteq \Hset \subseteq
\{x\in\R^n:\norm{x-c(p)}_2 \le \rho(p)\}.
\label{eq:ballcontain}
\end{equation}
Consequently,
\begin{equation}
\diam(\X) \le 2\rho_\star,
\qquad
\rho_\star \doteq \min_{p\in\Delta_m}\rho(p).
\label{eq:rhostar}
\end{equation}
\end{proposition}

\begin{proof}
Fix $x\in\Hset$ and $p\in\Delta_m$.
Since
$\norm{x-a^{(i)}}_2^2\le \xi_i^+$ for all $i$,
\[
\sum_{i=1}^m p_i \norm{x-a^{(i)}}_2^2
\le
\sum_{i=1}^m p_i \xi_i^+.
\]
Using the weighted variance identity around $c(p)=Ap$,
\[
\sum_{i=1}^m p_i \norm{x-a^{(i)}}_2^2
=
\norm{x-c(p)}_2^2 +
\sum_{i=1}^m p_i\norm{a^{(i)}-c(p)}_2^2,
\]
which yields \eqref{eq:ballcontain}.
The pairwise-distance expression in \eqref{eq:rhop} is the standard identity for weighted scatter. Minimizing over $p\in\Delta_m$ gives \eqref{eq:rhostar}.
\end{proof}

\begin{remark}
\label{rem:special}
Two explicit choices are useful. For the uniform weight
$p=\frac{1}{m}\one$,
\begin{equation}
\rho(p)^2
=
\frac{1}{m}\sum_{i=1}^m \xi_i^+
-\frac{1}{m}\operatorname{tr}S(A).
\label{eq:uniformrho}
\end{equation}
Hence
\[
\diam(\X)\le
2\sqrt{\frac{1}{m}\sum_{i=1}^m \xi_i^+
-\frac{1}{m}\operatorname{tr}S(A)}.
\]
For the pairwise choice
$p=\tfrac{1}{2}(e_i+e_j)$,
\begin{equation}
\rho(p)^2 =
\frac{\xi_i^+ + \xi_j^+}{2}
-\frac{\norm{a^{(i)}-a^{(j)}}_2^2}{4},
\label{eq:pairrho}
\end{equation}
which gives a simple two-anchor lens bound.
\end{remark}

\begin{proposition}
\label{prop:support}
Assume $\Hset$ has nonempty interior and let
\[
h_{\Hset}(v) \doteq \max_{x\in\Hset} v\tran x
\]
denote the support function of $\Hset$.
Then, for every
$v\in\R^n$,
\begin{equation}
h_{\Hset}(v)
=
\min_{p\in\Delta_m}
\bigl(v\tran Ap + \norm{v}_2\,\rho(p)\bigr).
\label{eq:support}
\end{equation}
Hence, for every unit vector $v$,
\begin{equation}
\widthop_v(\X)
\le \psi_{\Hset}(v)
\doteq h_{\Hset}(v)+h_{\Hset}(-v),
\label{eq:psi}
\end{equation}
and, if $S(A)\succ0$,
\begin{equation}
\widthop_v(\X)
\le
\min\!\left\{
\frac{1}{2}w\tran\abs{Q_mA\tran S(A)^{-1}v},
\ \psi_{\Hset}(v)
\right\}.
\label{eq:hybridwidth}
\end{equation}
\end{proposition}

\begin{proof}
Consider the SOCP
\[
\max_x\{v\tran x:
\norm{x-a^{(i)}}_2^2\le \xi_i^+,\ i=1,\ldots,m\}.
\]
Its Lagrangian dual is
\[
\min_{\mu\ge0,\ t=\one\tran\mu>0}
\left\{
\mu\tran(\xi^+-\omega)
+\frac{1}{4t}\norm{v+2A\mu}_2^2
\right\}.
\]
Under the Slater assumption for $\Hset$, strong duality holds.
Writing
$\mu=t p$ with $p\in\Delta_m$ gives
\[
v\tran Ap + \frac{\norm{v}_2^2}{4t} + t\rho(p)^2.
\]
Minimizing over $t>0$ yields \eqref{eq:support}, with minimum value
$v\tran Ap+\norm{v}_2\rho(p)$.
Since
$\widthop_v(\Hset)=h_{\Hset}(v)+h_{\Hset}(-v)$
and $\X\subseteq\Hset$, \eqref{eq:psi} follows. Combining \eqref{eq:psi} with Proposition~\ref{prop:dirbound} gives \eqref{eq:hybridwidth}.
\end{proof}

\begin{remark}
\label{rem:beta}
Formula \eqref{eq:support} reduces the support computation to a low-dimensional simplex-constrained optimization;
in practice one may also solve the equivalent primal SOCP directly.
A simpler closed-form surrogate follows from the support of each
individual ball:
\[
h_{\Hset}(v) \le \min_i\{v\tran a^{(i)} + \sqrt{\xi_i^+}\}.
\]
Hence
\begin{equation}
\begin{aligned}
\psi_{\Hset}(v)
&\le
\beta_{\Hset}(v) \\
&\doteq
\min_i\{v\tran a^{(i)}+\sqrt{\xi_i^+}\} \\
&\quad -
\max_i\{v\tran a^{(i)}-\sqrt{\xi_i^+}\}.
\end{aligned}
\label{eq:betaH}
\end{equation}
\end{remark}

\begin{corollary}
\label{cor:hybridbox}
Assume $S(A)\succ0$ and let
$S(A)=U\Lambda U\tran$
with eigenvectors $u_1,\ldots,u_n$ and eigenvalues
$\lambda_1,\ldots,\lambda_n>0$. Define
\begin{equation}
b_j \doteq
\min\!\left\{
\frac{\norm{w}_2}{2\sqrt{\lambda_j}},
\ \psi_{\Hset}(u_j)
\right\},
\qquad j=1,\ldots,n.
\label{eq:bj}
\end{equation}
Then $\X$ is contained in the orthogonal box aligned with $U$
and side lengths $b_j$, hence
\begin{equation}
\diam(\X) \le \Big(\sum_{j=1}^n b_j^2\Big)^{1/2},
\qquad
\vol(B_U) \le \prod_{j=1}^n b_j.
\label{eq:hybridbox}
\end{equation}
Discarding the $\psi_{\Hset}$ term recovers
\begin{equation}
\vol(B_U)
\le
\left(\frac{\norm{w}_2}{2}\right)^n
\frac{1}{\sqrt{\det S(A)}}.
\label{eq:detbound}
\end{equation}
\end{corollary}

The bounds above reveal two distinct mechanisms.
The matrix
$S(A)$ quantifies how the difference equations enlarge the
polyhedral outer set $\Xd$, whereas the $\Hset$-induced terms
$\rho_\star$ and $\psi_{\Hset}$ quantify how the intersection with
the upper balls truncates that enlargement.
In particular,
Proposition~\ref{prop:ball} remains finite whenever $\Hset$ is
nonempty, even when the $\Xd$-based diameter bound in
Corollary~\ref{cor:diam} deteriorates severely.
\section{Geometry-based anchor selection}
Consider a pool of candidate anchors
$\{b^{(1)},\ldots,b^{(N)}\}$ and suppose that only $k<N$ of them can be used.
For a subset
$\mathcal{I}\subset\{1,\ldots,N\}$ of cardinality $k$, let
$A_{\mathcal{I}}$ be the corresponding anchor matrix and let
$S_{\mathcal{I}} = A_{\mathcal{I}}Q_kA_{\mathcal{I}}\tran$.
The previous section yields two natural \emph{offline design scores}:
\begin{equation}
\gamma_E(\mathcal{I}) \doteq \lambda_{\min}(S_{\mathcal{I}}),
\qquad
\gamma_D(\mathcal{I}) \doteq \det S_{\mathcal{I}}.
\label{eq:EDscores}
\end{equation}
These scores are meaningful only for affinely spanning subsets: if $S_{\mathcal{I}}$ is rank deficient, then $\gamma_E(\mathcal{I})=\gamma_D(\mathcal{I})=0$ and the associated $\Xd$-based directional certificates are not finite.
When all interval widths are equal and the subset size $k$ is fixed,
maximizing $\gamma_E$ minimizes the $\Xd$-based diameter bound
\eqref{eq:diam}, whereas maximizing $\gamma_D$ minimizes the
polyhedral principal-axis volume bound \eqref{eq:detbound}.
If the
widths are not uniform, the corresponding weighted scores become
\begin{equation}
J_E(\mathcal{I}) \doteq
\frac{\|w_{\mathcal{I}}\|_2^2}{\lambda_{\min}(S_{\mathcal{I}})},
\qquad
J_D(\mathcal{I}) \doteq
\frac{\|w_{\mathcal{I}}\|_2^{2n}}{\det S_{\mathcal{I}}}.
\label{eq:weighted_scores}
\end{equation}
The notation E/D is reminiscent  of the stochastic 
setting, but here  $\lambda_{\min}(S_{\mathcal{I}})$ and $\det S_{\mathcal{I}}$ act as deterministic E- and D-optimality criteria: they bound guaranteed-set surrogates rather than local covariance matrices.
The quantities $\rho_\star$ and $\psi_{\Hset}$ remain measurement-dependent and are therefore online certificates for an already selected subset.
Likewise, the weighted scores in \eqref{eq:weighted_scores} are genuinely offline only when conservative width envelopes $w_{\mathcal{I}}$ are known a priori;
otherwise \eqref{eq:EDscores} are the natural pre-measurement selectors. The E-score targets worst-direction ambiguity, whereas the D-score reflects overall spread.
In moderate-size pools both can be optimized exhaustively; for larger pools they can be embedded in greedy or branch-and-bound schemes.

\begin{remark}[Scalability and Combinatorial Search]\rm 
Exact offline anchor selection is inherently a combinatorial problem, requiring the evaluation of $\binom{N}{k}$ subsets. However, this is rarely a limiting factor in practice for two reasons. First, in standard applications like GPS or drone localization, the total number of available anchors $N$ is typically small (for example, in the range 4 to 12 in GPS localization), making exhaustive search computationally viable. Second, for larger pools, suboptimal greedy addition or branch-and-bound schemes can efficiently approximate the E- or D-scores. Furthermore, it is important to note that the \emph{online} certification phase (evaluating $\rho_\star$ and $\psi_{\Hset}$) is strictly non-combinatorial; it simply evaluates the properties of the single $k$-subset that has been selected in the offline phase.
\end{remark}

\section{Numerical validation}
Experiments in this section are performed  in $\R^2$ and use a squared-range interval model
\begin{equation}
\xi_i \in [d_i^2-\delta,\, d_i^2+\delta],
\qquad
d_i = \norm{x^\star-a^{(i)}}_2,
\label{eq:numinterval}
\end{equation}
with $\delta=0.05$, i.e., every interval width is equal to $0.1$.
All computations below use the actual localization set
$\X=\Xd\cap\Hset$.
Since $\X$ is convex and is described by the linear
inequalities \eqref{eq:Xdhalf} together with the ball constraints
\eqref{eq:Hset}, its support in any direction $v$ is obtained from the
convex program
\[
\max_x\{v\tran x : x\in\Xd,\ \norm{x-a^{(i)}}_2 \le \sqrt{\xi_i^+},\ i=1,\ldots,m\},
\]
which is a small convex second-order cone program (SOCP).
Therefore the exact coordinate-aligned bounding
box of $\X$ is obtained by solving, for each coordinate vector $e_j$, the
pair of problems $\max e_j\tran x$ and $\min e_j\tran x$ over $\X$.
In $\R^2$ this means four convex programs. While the numerical experiments in this section are visually presented in $\R^2$, the theoretical results and methodology are not restricted to $n=2$ and scale seamlessly to 3D or generic $\R^n$. For a 3D application (such as UAV/drone localization), computing the exact coordinate box requires solving exactly $6$ small convex programs, and generally $2n$ such problems in $\R^n$.

\subsection{Anchor subset selection as offline design}
We first test the proposed geometry scores on a combinatorial
selection problem.
In each Monte Carlo trial, $N=8$ candidate anchors are sampled
uniformly in the annulus of radii $1$ and $2$ around the target, and
$k=4$ anchors must be selected.
Since $\binom{8}{4}=70$, we enumerate
all subsets. For every subset we compute the exact area of the
coordinate-aligned outer box of $\X$, together with the D-score
$\gamma_D=\det S$ and the E-score
$\gamma_E=\lambda_{\min}(S)$ from \eqref{eq:EDscores}.
We compare four policies
over $60$ trials: the oracle subset minimizing the exact box area, the
subset maximizing $\gamma_D$, the subset maximizing $\gamma_E$, and
the mean area over all subsets.
Table~\ref{tab:selection} summarizes the results. Both geometry rules
are close to the oracle, with the determinant-based score performing
slightly better on this area-oriented metric: its mean box area is only
$4.1\%$ above the oracle, while the E-score is $6.2\%$ above it.
Both
rules are dramatically better than an average subset, whose mean box
area is about four times larger than the oracle.
In $25$ of the
$60$ trials the D-score recovers the oracle subset exactly; the E-score
does so in $24$ trials.
Figure~\ref{fig:cdf} shows that both scores are
within $2\%$ of the oracle in nearly half of the trials.
For the \emph{online} certificate, let
\[
\begin{aligned}
B_{d,j} &\doteq \frac{1}{2}w\tran\big|Q_mA\tran S(A)^{-1}e_j\big|,\\
B_{\Hset,j} &\doteq \psi_{\Hset}(e_j), \qquad j=1,2.
\end{aligned}
\]
Because $\X\subseteq\Xd$ and $\X\subseteq\Hset$, the exact coordinate
width of $\X$ along axis $e_j$ is no larger than either quantity.
We
therefore form the hybrid coordinate-box bound
\[
B_{\mathrm{hyb}} \doteq \prod_{j=1}^2 \min\{B_{d,j},B_{\Hset,j}\}.
\]
In each coordinate, we start from the $\Xd$-based width estimate and, if
the $\Hset$-based width is smaller, we replace it by that tighter value.
Figure~\ref{fig:scatter} plots $B_{\mathrm{hyb}}$ against the exact area
of the coordinate-aligned box of $\X$.
Both axes are logarithmic because
the values span several orders of magnitude across the $4200$ subsets.
The reported log-log correlation is simply the Pearson correlation
between $\log_{10} B_{\mathrm{hyb}}$ and the logarithm of the exact box
area;
values close to one mean that the bound correctly tracks which
subsets produce small and large uncertainty boxes.
Here the correlation
is $0.984$, while the mean ratio $B_{\mathrm{hyb}}/A_{\mathrm{box}}(\X)$ is
$1.69$ (median $1.50$).
Thus, the new $\Hset$-aware bound is not only
valid, but also predictive of the realized size of the localization set.
\begin{table}[t]
\centering
\caption{Subset selection over $60$ Monte Carlo trials ($N=8$, $k=4$).}
\label{tab:selection}
\footnotesize
\begin{tabular}{@{}lcc@{}}
\toprule
Method & Mean box area of $\X$ & Rel.
to oracle \\
\midrule
Oracle subset & $1.42\times 10^{-3}$ & $1.00$ \\
D-score $\max \det S$ & $1.47\times 10^{-3}$ & $1.04$ \\
E-score $\max \lambda_{\min}(S)$ & $1.50\times 10^{-3}$ & $1.06$ \\
Average subset & $5.85\times 10^{-3}$ & $4.13$ \\
\bottomrule
\end{tabular}
\normalsize
\end{table}

\begin{figure}[t]
\centering
\includegraphics[width=0.97\linewidth]{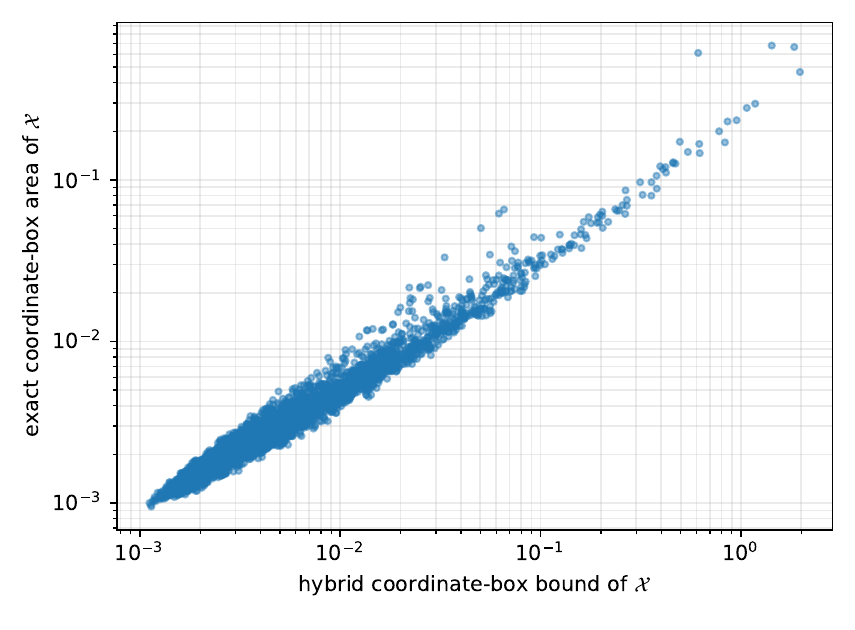}
\caption{Hybrid coordinate-box bound $B_{\mathrm{hyb}}$ versus the exact coordinate-box area of the actual localization set $\X$ over all subsets in the Monte Carlo experiment.
Both axes are on logarithmic scale.}
\label{fig:scatter}
\end{figure}

\begin{figure}[t]
\centering
\includegraphics[width=0.97\linewidth]{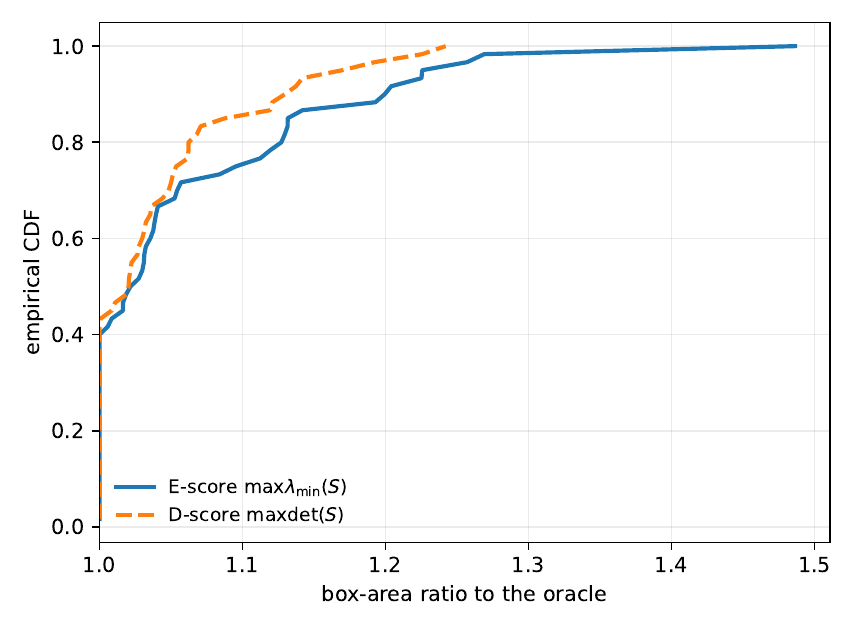}
\caption{Empirical CDF of the box-area ratio to the oracle over the $60$ Monte Carlo trials.
The D-score slightly dominates the E-score on this area-based metric and is therefore the practical default, while both remain close to the oracle.}
\label{fig:cdf}
\end{figure}

\subsection{Why $\Hset$-aware online certification is needed}
To isolate the role of the $\Hset$-aware terms, we next consider three
anchors at $(-1,0)$, $(1,0)$, and $(0,h)$,
with true target $x^\star=(0.15,0.35)$.
The height $h$ is varied from
$1.2$ down to $0.08$, so that the third anchor approaches the line
through the first two.
Figure~\ref{fig:sweep} summarizes the results.

The left panel reports the exact vertical width of $\X$, obtained from
the two support problems $\max e_2\tran x$ and $\min e_2\tran x$ over
$\X$, together with the $\Xd$-based width bound from
Proposition~\ref{prop:dirbound}, the $\Hset$-support bound
$\psi_{\Hset}(e_2)$ from Proposition~\ref{prop:support}, and the hybrid
minimum of the two.
As $h$ decreases from $1.2$ to $0.08$, the exact
vertical width grows
from $7.27\times 10^{-2}$ to $6.86\times 10^{-1}$, while the
$\Xd$-based bound grows from $8.33\times 10^{-2}$ to $1.25$.
The $\Hset$-support bound remains much tighter near degeneracy,
changing from $1.00\times10^{-1}$ to $7.04\times10^{-1}$;
the hybrid
bound therefore follows the polyhedral term in well-spread geometries
and automatically switches to the ball-intersection term when the
anchor set becomes nearly collinear.
The right panel compares the pure polyhedral diameter bound
\eqref{eq:diam} with the ball-induced bound $2\rho_\star$ from
Proposition~\ref{prop:ball}.
Over the same sweep, the former changes
from $8.84\times10^{-2}$ to $1.33$, whereas the latter changes only
from $4.47\times10^{-1}$ to $7.18\times10^{-1}$.
This is precisely the
gap between $\diam(\Xd)$ and the actual size of $\X$: the conditioning
of $S(A)$ drives the growth of the loose polyhedral envelope, while the
intersection with the upper balls still clips the feasible region
substantially.
\begin{figure}[t]
\centering
\includegraphics[width=0.99\linewidth]{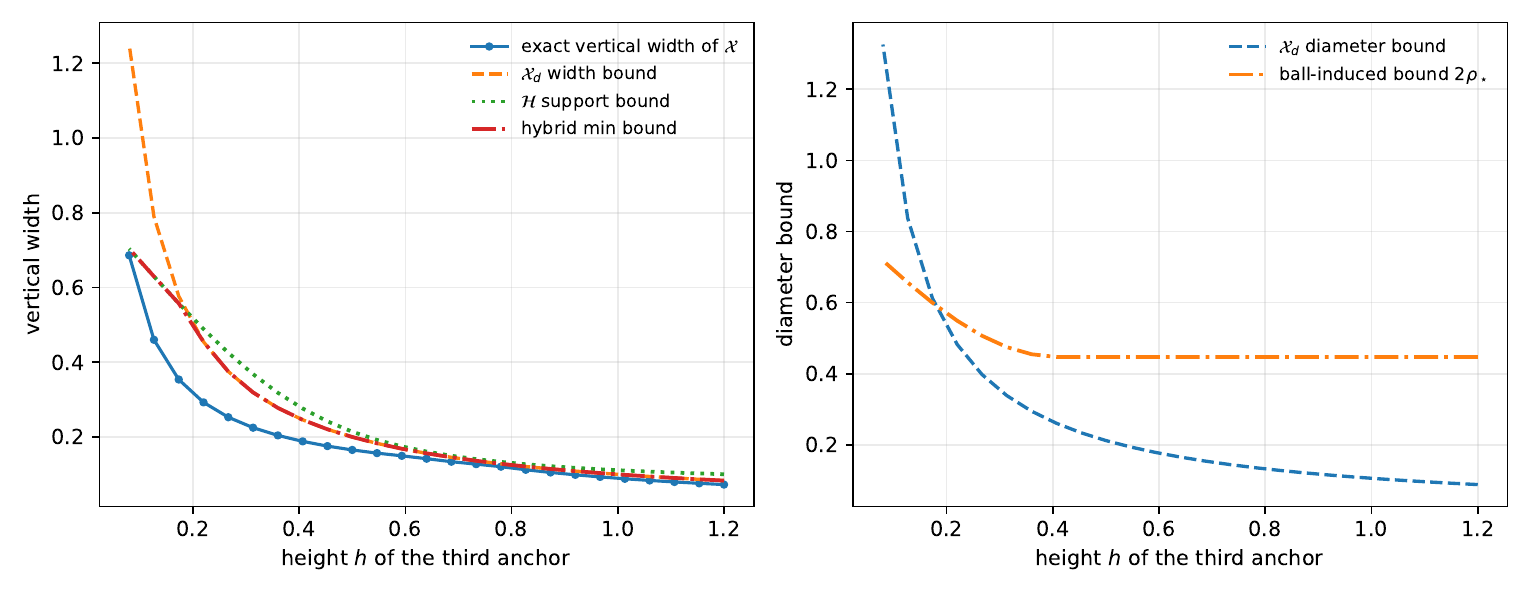}
\caption{Geometry sweep for the actual localization set $\X$. Left: exact vertical width, the $\Xd$-based bound, the $\Hset$-support bound, and their hybrid minimum.
Right: pure polyhedral diameter bound and ball-induced diameter bound.}
\label{fig:sweep}
\end{figure}

\section{Design perspective and practical implications}
The new results support a simple two-stage workflow for set-membership
multilateration.
Before measurements are collected, one uses the geometry-only
scores $\gamma_E$ and $\gamma_D$ in \eqref{eq:EDscores} to select or schedule a
subset of anchors.
These scores are inexpensive, target-free, and depend only
on anchor positions, so they are natural design surrogates in applications with
communication, energy, or sensing-budget constraints.
After measurements are
available, one keeps the selected subset and computes the corresponding
$\Hset$-aware certificates $\rho_\star$ and $\psi_{\Hset}$ to assess the actual
size of the resulting localization set $\X$.
This separation also clarifies the respective roles of the D- and E-scores.
If one seeks a compact overall uncertainty region, then maximizing $\det S$ is
the appropriate default, as confirmed by the subset-selection experiment on the
exact box area of $\X$.
If the application is instead driven by worst-direction
ambiguity, then maximizing $\lambda_{\min}(S)$ is the natural alternative because
it directly minimizes the dominant term in \eqref{eq:diam}.
The online
$\Hset$-aware terms should then be viewed as post-selection certificates: they
do not replace the geometry scores, but they substantially sharpen the realized
uncertainty assessment of the selected subset, especially in nearly collinear
or clustered geometries.
\section{Conclusions}
We re-framed the localization-set construction of \cite{CalafioreSIAM}
as a design problem for anchor subset selection under bounded range
errors.
The matrix $S(A)=AQ_mA\tran$ yields simple geometry-only
E- and D-type scores that can be used offline to select anchors before
measurements are collected.
At the same time, the special structure
$\X=\Xd\cap\Hset$ yields new $\Hset$-aware online certificates for the
selected subset: the simplex-aggregated radius $\rho_\star$ gives a
finite diameter bound even when the $\Xd$-based certificate deteriorates
badly, and the support-function formula for $\Hset$ leads to hybrid
directional bounds that track the actual size of $\X$ closely.
The reported experiments show near-oracle subset selection on the tested instances and indicate that the new $\Hset$-aware certificates explain a substantial part of the gap between the loose polyhedral envelope and the realized localization set.

Future research directions include extending these geometry-aware certificates to dynamic localization, where the target's state is tracked over time using bounded-error kinematic models. Additionally, incorporating unknown-but-bounded uncertainty on the anchor positions themselves, and developing decentralized algorithms for evaluating the D- and E-scores in large-scale networks, remain open and promising problems.


\begin{thebibliography}{99}
\bibitem{CalafioreSIAM}
G.~C. Calafiore,
\newblock Set-membership localization via range measurements,
\newblock \emph{SIAM Journal on Optimization}, to appear.
Preprint available at \url{http://arxiv.org/abs/2603.04867}

\bibitem{Milanese1996}
M.~Milanese, J.~Norton, H.~Piet-Lahanier, and E.~Walter, editors,
\newblock \emph{Bounding Approaches to System Identification},
\newblock Plenum Press, New York, 1996.

\bibitem{Jaulin2001}
L.~Jaulin, M.~Kieffer, O.~Didrit, and E.~Walter,
\newblock \emph{Applied Interval Analysis},
\newblock Springer, London, 2001.

\bibitem{ChanHo1994}
Y.~T.
Chan and K.~C. Ho,
\newblock A simple and efficient estimator for hyperbolic location,
\newblock \emph{IEEE Transactions on Signal Processing}, 42(8):1905--1915, 1994.

\bibitem{BeckStoicaLi2008}
A.~Beck, P.~Stoica, and J.~Li,
\newblock Exact and approximate solutions of source localization problems,
\newblock \emph{IEEE Transactions on Signal Processing}, 56(5):1770--1778, 2008.

\bibitem{Biswas2006}
P.~Biswas, T.-C.
Liang, K.-C. Toh, T.-C. Wang, and Y.~Ye,
\newblock Semidefinite programming approaches for sensor network localization with noisy distance measurements,
\newblock \emph{IEEE Transactions on Automation Science and Engineering}, 3(4):360--371, 2006.

\bibitem{Tseng2007}
P.~Tseng,
\newblock Second-order cone programming relaxation of sensor network localization,
\newblock \emph{SIAM Journal on Optimization}, 18(1):156--185, 2007.

\bibitem{Bancroft1985}
S.~Bancroft,
\newblock An algebraic solution of the GPS equations,
\newblock \emph{IEEE Transactions on Aerospace and Electronic Systems}, AES-21(1):56--59, 1985.

\bibitem{NaddafzadehShirazi2014}
G.~Naddafzadeh~Shirazi, M.~B.
Shenouda, and L.~Lampe,
\newblock Second order cone programming for sensor network localization with anchor position uncertainty,
\newblock \emph{IEEE Transactions on Wireless Communications}, 13(2):749--763, 2014.

\bibitem{ShiMaoAndersonYangChen2017}
X.~Shi, G.~Mao, B.~D.~O.
Anderson, Z.~Yang, and J.~Chen,
\newblock Robust localization using range measurements with unknown and bounded errors,
\newblock \emph{IEEE Transactions on Wireless Communications}, 16(6):4065--4078, 2017.

\bibitem{ChangSahai2006}
C.~Chang and A.~Sahai,
\newblock Cram\'er--Rao-type bounds for localization,
\newblock \emph{EURASIP Journal on Advances in Signal Processing}, 2006:94287, 2006.

\bibitem{Bishop2010}
A.~N.
Bishop, B.~Fidan, B.~D.~O. Anderson, K.~Dogancay, and P.~N. Pathirana,
\newblock Optimality analysis of sensor-target localization geometries,
\newblock \emph{Automatica}, 46(3):479--492, 2010.

\bibitem{WiddisonLong2024}
E.~Widdison and D.~G.
Long,
\newblock A review of linear multilateration techniques and applications,
\newblock \emph{IEEE Access}, 12:26251--26266, 2024.

\bibitem{MorenoSalinas2013}
D.~Moreno-Salinas, A.~M.
Pascoal, and J.~Aranda,
\newblock Optimal sensor placement for multiple target positioning with range-only measurements in two-dimensional scenarios,
\newblock \emph{Sensors}, 13(8):10674--10710, 2013.
\end{thebibliography}
\end{document}